\documentclass{amsart}
\usepackage{amsmath,amsthm,amsfonts,amssymb,amstext, latexsym, amscd, color,  graphics, graphicx}

\newtheorem{prop}{Proposition}

\begin{document}

\title[The free-choice paradigm]{On the mathematics of the free-choice paradigm}
\author{Peter Selinger}
\author{Kristopher Tapp}

\address{Department of Mathematics and Statistics\\Dalhousie University\\Halifax, Nova Scotia B3H 4R2, Canada}
\email{selinger@mathstat.dal.ca}
\address{Department of Mathematics\\ Saint Joseph's University\\
5600 City Ave.\\
Philadelphia, PA 19131}
\email{ktapp@sju.edu}

\maketitle
\date{\today}

\section{Introduction}

Experimental design can be tricky to get right.  A very potent illustration of this point is found in Chen and Risen's recent identification of a logical flaw in a number of past \emph{free choice} experiments studying the psychological concept of \emph{cognitive dissonance}\cite{CR1}.  This illustration is unique in many ways.  First, the mistake is subtle yet elementary; in the simplest of the affected experiments, the mistake is equivalent to the error of misunderstanding the Monty Hall problem.  Second, the mistake affected a fairly large number of research experiments performed over a span of five decades.  Third, it challenged some of the experimental evidence for dissonance theory, which is a well known and celebrated area of social psychology.  In fact, the term ``cognitive dissonance'' has migrated from the scientific realm into popular culture, where it occurs frequently in New York Times articles, is the title of a Pod Cast, is featured in several Dilbert cartoons, and has been used to explain everything from why President Clinton was not impeached to why some weight loss programs work better than others.

In this paper, our first goal is expository -- we wish to tell this interesting story, which seems to have largely escaped the notice of the mathematics community.  We will describe the theory of cognitive dissonance, the experimental evidence on which it rests, and the flaw observed by Chen and Risen.  Our second goal is to propose methods by which this mistake can be fixed.  Chen and Risen already implemented a modified experiment that used a novel type of control group to avoid the logical pitfall.  We will describe experimental designs by which the flaw can be avoided even without a control group.  Our methods are based on the inherent symmetry of a free-choice experiment.  Our third goal is to describe additional problems with all of these free choice experiments (including our own methods and the Chen-Risen method), calling back into question whether any free choice experiment can correctly measure the effects of cognitive dissonance.
\section*{Acknowledgements} The authors would like to thank Keith Chen and Jane Risen for their insightful suggestions and comments on this work.

\section{Cognitive Dissonance}
Festinger coined the term \emph{cognitive dissonance} to describe the uncomfortable cognitive state that arises when one's actions are inconsistent with one's underlying attitudes/beliefs~\cite{F}.  Dissonance Theory is largely about our tendency to reduce dissonance by shifting our underlying attitudes.  For example, after you break up with a romantic partner, you might tell yourself that you never really liked him/her in the first place.  That's dissonance -- you are shifting your attitude (how you really feel about your ex-partner) to make it more compatible with your action (breaking up with him/her).  It helps you to avoid regret and move on.

If the end of a relationship seems too mundane, dissonance theory began with the end of the world.  In 1955, Festinger learned of a cult, which he named \emph{The Seekers}, who believed that God would destroy the earth on December 21 of that year, but aliens from the planet Clarion would arrive in a spaceship before the destruction to save the true believers.  They really believed it.  The gave away their belongings, divorced their non-believing spouses, and sacrificed all that they had in order to follow the specific instructions they had received from Clarion.  One of Festinger's students infiltrated the group and was able to observe their preparations for the earth's annihilation~\cite{Clarion}.  They gathered to wait with excited anticipation for the spaceship that would save them at midnight.  And they waited.  A few minutes after midnight, they all reset their watches to agree with the one member whose watch still read before midnight.  Several minutes later, when even the modified watches read past midnight, one member realized that he had not followed Clarion's instruction to remove all metal objects in preparation for space flight.  He still had a metal tooth filling.  He removed it.  After several more hours of nervous waiting, there was still no spaceship and no annihilation.  How did they handle the gaping inconsistency they now felt between their actions and their new resignation that the earth would continue to spin?  Cognitive dissonance doesn't get much bigger than that.  Festinger and his students had predicted that they would employ heroic measures to rationalize away the apparent inconsistency.  And so they did.  At 4:00am, they received a final message from Clarion: ``This little group, sitting all night long, has spread so much goodness and light that the God of the Universe has spared the Earth from its destruction.''  It was their group that had saved the earth!  As Festinger had predicted, The Seekers not only found a new consistency-restoring belief to cling to, but they used every possible media outlet to share their new belief with the world that they had saved.

Thus began a psychology theory which ``has had an amazing fifty-year run.''~\cite{Cooper}  Dissonance theory rests on the results of hundreds of laboratory experiments, which fall into three experimental paradigms: \emph{free-choice}, \emph{induced compliance}, and \emph{effort justification}.  The error observed by Chen and Risen affects only the free-choice paradigm, which we will discuss in the next section.  In the remainder of this section, we will briefly discuss the two unaffected paradigms.

The first \emph{induced compliance} experiment was performed by Festinger and Carlsmith~\cite{FC}, using students as subjects.  In their experiment, each subject was asked to perform a boring task involving turning pegs on a peg board, and then asked to convince another student that the task had been interesting (in an elaborate ruse to convince the subject to tell this lie, s/he was told something like this: ``The other student is waiting to be the next subject in this experiment.  Unlike you, he is in the experimental group, which means that before performing the peg board task, he will be prepped by a confederate to believe that the task will be interesting.  But the confederate did not show up, so can I hire you to play the role of the confederate?'')  Finally, the subject answered debriefing questions, including questions about how interesting the task was.  The result was that subjects paid $\$1$ to lie generally came to actually believe themselves that the task had been interesting, while subjects paid $\$20$ to lie did not.  Presumably, cognitive dissonance was aroused by the inconsistency between finding the task boring but saying that it was interesting.  A student paid $\$20$ had a rational explanation for this inconsistency (it's worth lying for that much money) while a student paid $\$1$ may have unconsciously reasoned something like this: ``I'm not dumb enough to lie for one dollar, so I must have been telling the truth.''  Dozens of other induced compliance experiments have been performed.  The majority of them used a similar design in which students were convinced (by a similarly elaborate ruse) to write counter-attitudinal essays.  It was generally found that students paid only a small amount to write such an essay tended to shift their beliefs in the direction of what they had written, while students paid a large amount did not.

The first \emph{effort justification} experiment was performed by Aronson and Mills~\cite{AM}.  They found that students who endured a severe initiation to join what turned out to be a dull club ended up liking the club more than students who endured only a mild initiation.  Presumably, these students had a larger degree of dissonance aroused by the inconsistency between their efforts and their group experience, and thus had stronger motivation to reduce dissonance by shifting their attitudes about the group.

The induced compliance and effort justification experiments, and the dissonance theory interpretation of their results, have been been attacked on many fronts. Cooper's book~\cite{Cooper} provides a wonderful overview of five decades of probing and re-working of these experiments and their interpretations.  Versions of the experiments have been performed on subjects wired to polygraphs, subjects who had unknowingly taken stimulants, subjects who had taken placebos but believed they had taken stimulants, subjects who underwent electric shock therapy, subjects with amnesia, children and draft-dodgers, and the outcomes have been correlated with the subjects' degree of self-esteem, level of introversion/extroversion, and ethnicity, to name a few.  Cooper argues that dissonance theory largely survived the attacks in a modified form.

None of the attacks were mathematical in nature until the free-choice paradigm was challenged on probabilistic grounds in 2010.

\section{The Free-Choice Paradigm}

The first free-choice experiment was performed by Brehm~\cite{B}, a student of Festinger.  Dozens more free-choice experiments followed.  In this section, we review the structure of a typical free-choice experiment.

A collection of objects is used, say 15 household objects.  In the first stage, a subject is asked to rank these objects from $1$=most desirable to $15$=least desirable.  In the second stage, the subject is asked to choose between two of these objects, and is usually told that s/he will be allowed to take the selected one home.  For example, the subject might be asked to choose between the objects that s/he just ranked $7^\text{th}$ best and $9^\text{th}$ best (where the numbers $7$ and $9$ are pre-selected and constant over all subjects).  Finally in the third stage, the subject is again asked to rank all 15 objects.  This process is repeated with many subjects.  Each subject's \emph{spread} is calculated as the amount that the ranking of the chosen object improves (i.e., decreases) plus the amount that the ranking of the rejected object worsens (i.e., increases) between the first and third stage of the experiment.

Positive spread is interpreted as an indication that the subject shifted his/her ranking to make it more consistent with his/her object choice, presumably to reduce dissonance. For example, if I chose to take home the hairdryer instead of the toaster, then dissonance is aroused by the inconsistency between my choice of the hairdryer and my lingering memories of the reasons that I almost chose the toaster (I already own a hairdryer and the toaster sure looks shiny).  I can reduce dissonance by downplaying these memories (the new hairdryer is better than the one that I own, and shiny things hurt my eyes).  This is the kind of opinion-shift that causes positive spread and was interpreted as evidence of \emph{choice-induced attitude change} resulting from dissonance.

In some experiments, positive average spread was taken as evidence for dissonance theory.  In other experiments, the spread was required to be more positive for an experimental group than for a control group.  Members of a control group might skip stage two, or else in stage two might be asked to make a choice between a pair of irrelevant objects.  Since they do not make the crucial choice in stage two, how is their spread defined?  Notice that in stage two, the majority of subjects in the experimental group might be expected to make the \emph{consistent} choice of the $7^\text{th}$ ranked object over the $9^\text{th}$; however, since $7$ and $9$ are fairly close together, a reasonable minority might make the \emph{reversal} choice of the $9^\text{th}$ over the $7^\text{th}$.  For a member of the control group, spread can be defined as if this member had made the consistent choice; that is, spread equals the amount that the ranking of the originally-$7^\text{th}$-ranked object decreases below 7 plus the amount that the ranking of the originally-$9^\text{th}$-ranked object increases above 9.  Brehm took this approach, and he compared the average spread of all of the control group to the average spread of only those members of the experimental group who had made a consistent choice (thereby assuring that spread was defined by the same formula for everyone).  Other experiments have used a different approach in which members of the control group in stage two are randomly given either their $7^\text{th}$ or $9^\text{th}$ ranked object, and their spread is then defined as if they had chosen the object they were given.

Another common variation is to ask the subject in stages one and three to provide \emph{ratings} of the objects rather than a \emph{ranking} of the objects.  For example, each of the 15 objects might be separately rated on a scale of 1 to 10.  When ratings are used, the choice in stage two is usually made between two objects that in stage one were rated similarly and highly (because difficult decisions are hypothesized to arouse more dissonance).

More recently, the free-choice setup was simplified to detect dissonance in monkeys and children, who are not capable of communicating a complete ranking of a collection of objects~\cite{ESB}.  In fact, simplifying the setup made its hidden mathematical error become more noticeable.  In the monkey version, the error became essentially equivalent to the Monty Hall problem, which finally led to its discovery.  We will restrict our attention to the experiments on adults, and we refer to~\cite{T} for the story of monkeys and Monty Hall.

\section{The Chen-Risen critique of free-choice experiments}
The outcomes of these free-choice experiments on adults -- positive average spread, or higher average spread for an experimental group than a control group -- were taken as evidence that subjects shifted their rankings to reduce dissonance.  Did you catch the mistake?  Chen and Risen pointed out that these outcomes are predicted by a null hypothesis model in which subjects never change their minds.  In their model, each subject has a never-changing \emph{true ranking} of the objects.  But random noise can cause the rankings that a subject provides in stages one and three to differ from this true ranking, and can cause the stage two choice to be inconsistent with this true ranking.  Under natural hypotheses on the distributions by which this random noise is modeled, Chen and Risen showed that many of the outcomes of free-choice experiments which have previously been taken as evidence for dissonance theory are actually predicted by their model.  Thus, the past experimental results do not provide evidence against the ``nobody changes their minds'' null hypothesis.

The key intuition behind Chen and Risen's observation is this: a subject's choice in stage two provides added probabilistic information about his/her true ranking.  For example, suppose that in stage one, I rank the toaster $7^{\text{th}}$ best and the hairdryer $9^{\text{th}}$ best.  At this point, your best guess is that I truly like the toaster $7^{\text{th}}$ best and the hairdryer $9^{\text{th}}$ best, although the truth could be otherwise due to random noise.  Now suppose that in stage two, I select the hairdryer.  At this point, how do you expect that I truly feel about toasters and hairdryers?  Your best guess is that my true feelings are a sort of average of the feelings I indicated in stages one and two.  After factoring in the stage two probabilistic information, you would guess that I truly like hairdryers more (and toasters less) than I indicated in stage one.  Thus, you should predict that my stage three ranking will shift in the direction of my stage two choice -- not because I will change my mind, but rather because my stage two choice indicated that my true feelings were always in this direction.

Chen and Risen went on to design and implement a free-choice experiment which avoided the error of its predecessors by using a novel type of control group; specifically, one whose members went through the same three stages, but with the order of the second and third stages reversed.  Thus, a member of the control group ranked the objects, then ranked the objects again, then chose between the ones which were ranked  $7^{\text{th}}$ and $9^{\text{th}}$ best in the first ranking.  Just as before, the subject's spread was defined as the amount that the chosen object improved plus the amount that the rejected object worsened between the two rankings.  Their null hypothesis (or at least a more precise formulation of it) predicts that the control group and experimental group should have the same average spread.  After all, if subjects are acting like computer programs spitting out random perturbations of their true rankings, then there is no causality between the stages, so the order of the stages is irrelevant.  But if cognitive dissonance really causes choice-induced attitude change, then the experimental group should have higher average spread than the control group.  When the numbers were crunched, their experimental group had only nominally higher average spread, which provided only weak support for dissonance theory.

The free-choice paradigm is one of the three legs of the tripod on which dissonance theory rests.  Chen and Risen's novel control group represents one possible way to repair this leg.  Their experiment could be re-implemented with more subjects in the hope of obtaining statistically significant results.  But experiments are expensive.  The purpose of our paper is to contribute to the conversation about how future free-choice experiments would best be conducted.  In the next section, we present alternatives to Chen and Risen's control group method.
\section{A free-choice experiment with no control group}\label{sec-2}
In this section, we demonstrate that it is possible to conduct a free-choice experiment even without a control group, in a manner that avoids the probabilistic error that Chen and Risen identified in past experiments.

We begin by more precisely describing a very general ``nobody changes their minds'' null hypothesis model.  Let $n$ denote the number of objects (the previous discussion used $n=15$).  Let $S(n)$ denote the set of all $n!$ possible orderings of the objects.  The rankings that a subject provides in stages one and three are samples from his/her \emph{ranking distribution}, $r:S(n)\rightarrow[0,1]$, defined so that for each $\sigma\in S(n)$, the value $r(\sigma)$ represents the probability that the subject will provide that ranking.  Notice that the same distribution is used for stages one and three, which captures the key idea that the subjects never change their minds.  Also notice that our model is more general than the Chen-Risen model, because we do not necessarily assume that the subject has a well defined ``true ranking'' from which the ranking distribution is obtained by adding random noise.  It is an arbitrary distribution, so the only requirement is that $\sum_{\sigma\in S(n)} r(\sigma)=1$.

In past free-choice experiments, all subjects have made their stage-two choices between the objects they just ranked in a pair of comparison positions (like 7 and 9), which was pre-selected and constant over all subjects.  We now suggest three modifications of this design that make the experiment immune to the critique of Chen and Risen.

\begin{prop}\label{P}Under the null hypothesis model, the expected average spread equals zero if the free-choice experiment is modified in any of these ways:
\begin{enumerate}
\item All subjects make their stage-two choices between the same pre-selected pair of comparison objects (like hairdryer and toaster).
\item Each subject makes his/her stage-two choice between the objects he/she just ranked in a pair of comparison positions (like 7 and 9) that is uniformly randomly chosen separately for each subject.
\item Each possible pair of comparison positions is used, one for each subject.  For example, with $n=15$ objects, there are 105 pairs of distinct numbers between 1 and 15, so this experiment requires exactly 105 subjects.
\end{enumerate}
\end{prop}
\begin{proof}
To prove the first claim, notice that the subject's two rankings (in stages one and three) are independent samples from the same distribution.  Thus, exchanging the order of these two rankings does not affect the expected spread.  But on the other hand, exchanging the order of the two rankings has the effect of multiplying the expected spread by $-1$; it follows that the expected spread for a single subject (under the null hypothesis model) can only be zero.

It follows that the expected spread for a single subject equals zero if the stage-two choice is made between a pair of comparison objects that is uniformly randomly chosen separately for each subject.  This is because it equals zero for every possible pair of objects that could be chosen.

To prove the second claim, notice that the pair of comparison objects will depend here on the following two independent random processes:
\begin{itemize}
\item The pair of ranking positions is chosen uniformly at random.
\item The subject's stage-one ranking is sampled from his/her ranking distribution.
\end{itemize}
Since these two processes are independent, their order is irrelevant. If we imagine that the subject's ranking is provided first, then it is easy to see that, for any fixed ranking s/he provides, uniformly randomly choosing a pair of positions from this fixed ranking is equivalent to uniformly randomly choosing a pair of objects.

Notice that the above arguments establish that the expected spread for a \emph{single} subject equals zero in experiments (1) and (2).  If all subjects are identical, then claim (3) follows immediately, since in this case the expected spread for a single subject in (2) equals the average of all subjects' expected spreads in (3).  With non-identical subjects, we justify claim (3) as follows, provided that the assignment of subjects to $(i,j)$-pairs is random.  Assume for simplicity that $n=15$, and consider the $105\times 105$ matrix describing the expected spread for every subject using every possible-$(i,j)$ pair.  This matrix has a row for each subject and a column for each $(i,j)$-pair.  Since the average of the entries of every row of this matrix equals zero, the expected average of $105$ entries chosen one from each row (according to a random permutation of the numbers $1\ldots105$ that assigns subjects to pairs) must also equal zero.
\end{proof}

Each of the three modifications above represents a method for conducting a free choice experiment without a control group.  In all three cases, the null hypothesis model predicts an expected average spread of zero, so a measured average spread significantly above zero would be evidence against the null hypothesis, and might reasonably be attributed to choice-induced attitude change.

\section{Balanced positive and negative spread}

Dissonance theory has suffered from confusion as to what the expected spread would be if subjects never changed their minds.  Let $i$ and $j$ denote the ranking positions used in stage two, with $1\leq i<j\leq n$, and define $\Delta=j-i$.  Many past experimenters have assumed that the expected spread should equal zero.  It was first noted in~\cite{B} that \emph{regression to the mean} can cause negative spread when $\Delta$ is large.  For an extreme example, suppose that $n=15$, $i=1$ and $j=15$.  In this case, it is almost certain that the subject will make a consistent choice, in which case it is impossible for the spread to be positive, and regression to the mean will likely cause it to be negative.

Assuming that each subject's ranking distribution is obtained from a true ranking by adding random noise in a manner that satisfies some natural hypotheses, Chen and Risen stated as a theorem that the expected spread is positive for every choice of $i,j$, but they mentioned that regression to the mean can invalidate their proof when $\Delta$ is large.

In a typical situation, one might be led to expect positive spread when $\Delta$ is small due to the probabilistic arguments of Chen and Risen, and negative spread when $\Delta$ is large due to regression to the mean, although the exact meaning of ``large'' and ``small'' might depend on the specific ranking distributions of the subjects.

Proposition~\ref{P}(3) implies that choices of $i,j$ for which the expected spread is positive must be balanced by choices for which it is negative.  In the remainder of this section, we illustrate this balance with a simulated example that attempts to model the type of random noise that affects human rankings in free-choice experiments.

Our simulation uses $n=12$ objects.  When asked to give a ranking of these 12 objects, our simulated subject begins with its true ranking and modifies it by repeating the following procedure: it flips a weighted coin, and if the coin lands heads, then it swaps the objects in a randomly chosen pair of adjacent positions.  It stops making changes when the coin first lands tails.  Thus, the volume of random noise is controlled by the coin's probability of landing heads, which we denote $p$.  Furthermore, the simulated subject's choosing algorithm is the one induced by its ranking algorithm.  In other words, when asked to choose between a pair of objects, it uses the above algorithm to rank all $12$ objects, and then selects the better-ranked one of the pair.

In Table~\ref{table2}, we have calculated the expected spreads (rounded to three decimals) for all $\{i,j\}$ possibilities  with $p=0.8$.  Prior to rounding, the 66 values in this table must sum to exactly zero by Proposition~\ref{P}. Coincidentally, the rounded values sum to exactly zero as well.

\begin{table}[h!]
\caption{Expected spreads for the 12-object experiment with $p=0.8$.}
\tiny{
$$\begin{array}{r|rrrrrrrrrrr}
 & i=1 & 2 & 3 & 4 & 5 & 6 & 7 & 8 & 9 & 10 & 11 \\\hline
j=1 & - &  &  &  &  & &  &  & &  &  \\
2 & .319 & - & & &  & &  &  & &  &  \\
3 & -.010 & .557 & - &  &  &  &  &  & &  &  \\
4 & -.251 & .247 & .661 & - &  &  & &  &  &  & \\
5 & -.389 & .051 & .346 & .694 & - &  & &  &  & & \\
6 & -.458 & -.057 & .154 & .376 & .702 & - &  &  & &  & \\
7 & -.492 & -.111 & .050 & .184 & .384 & .704 & - & & &  & \\
8 & -.508 & -.138 & -.004 & .079 & .190 & .384 & .702 & - & &  &\\
9 & -.523 & -.157 & -.036 & .019 & .079 & .184 & .376 & .694 & - &  & \\
10 & -.557 & -.193 & -.078 & -.036 & -.004 & .050 & .154 & .346 & .661 & - &\\
11 & -.669 & -.306 & -.193 & -.157 & -.138 & -.111 & -.057
   & .051 & .247 & .557 & -\\
12 & -1.031 & -.669 & -.557 & -.523 & -.508 & -.492 & -.458
   & -.389 & -.251 & -.010 & .319 \\
\end{array}$$
}
\label{table2}\end{table}

In summary, this example is intended to illustrate Proposition~\ref{P} in order to aid our understanding of a question about which there has been some confusion in the literature, namely, what spread one should expect under the assumption that subjects never change their minds.

\section{Comparison of free-choice experiments}
In this section, we compare several versions of the free-choice experiment.  Some of our claims are justified by computer simulations of these experiments, in which subjects' random noise is modeled by the coin flipping process described in the previous section.  The three control-group-free experiments enumerated in Proposition~\ref{P} will be denoted as E1, E2 and E3 respectively.  The control-group-experiment of Chen and Risen will be denoted as E0.

As we have already pointed out, all four experiments are expected to detect no dissonance under the null hypothesis. This follows from Chen and Risen's result (for E0) and from Proposition~\ref{P} (for E1--E3). However, the experiments may differ in \emph{sensitivity}: how likely is each experiment to yield statistically significant evidence of dissonance in case the null hypothesis is false? Since this likelihood depends on the number of test subjects, to get a fair comparison, we assume that all four experiments use the same number of subjects.  For simplicity, we assume that all four experiments use 15 objects and 105 subjects.  In E0, these subjects are split between the experimental and control group.

One disadvantage of E0 comes from having half as many subjects in its experimental group (which increases the standard error by a factor of $\sqrt{2}$) and
from needing to calculate the spread \emph{difference} between two groups (which increases it by another factor of $\sqrt{2}$, assuming the two groups have similar standard deviations, which was the case in simulations).

The primary disadvantage of E2 is that many subjects are wasted on $(i,j)$-choices for which $\Delta$ is large enough that dissonance researchers would not hypothesize any spread due to choice-induced attitude change.  Dissonance and attitude change are only theorized to emerge when the decision is hard enough that the subject feels a need to rationalize choosing one item over the other.

In computer simulations, E0 was more likely than E2 to produce statistically significant evidence of a choice-induced spread phenomenon, if this phenomenon was modeled to only appear when $\Delta$ was small.  However, the winner depended completely on how we modeled the dependence on $\Delta$ of the choice-induced spread phenomenon, which felt like a fairly arbitrary decision.  The decision could not be empirically guided, since there is still no valid experimental evidence in the literature of \emph{any} type of choice-induced spread.

It might be possible to implement E1 so that it has the advantage but not the disadvantage of E2, by pre-selecting a pair of comparison objects that one expects to appear close together near the center of most subjects' rankings.  To achieve this, one could let the 15 items be ranked by a small number of test subjects, and then pick two items (e.g. toaster and hairdryer) that typically rank ``near the middle.'' The main experiment could then be done (with a different group of test subjects), always computing the toaster/hairdryer spread.  If it bears out that subjects in the main experiment also typically rank toaster and hairdryer close together, then not many subjects will be wasted on $(i,j)$-choices for which $\Delta$ is large, and no subjects are wasted on a control group either.

One might think that E3 is an improvement on E2, since it eliminates the random selection of $(i,j)$-pairs, and thereby eliminates the worry that E2's outcome could be blamed on a non-representative sampling of pairs.  However, in simulations, E2 and E3 had similar standard deviations.  Moreover, since the 105 subjects in E3 perform \emph{different} tasks, it is not necessarily valid to use the usual standard error formula, $SE=\sigma/\sqrt{105}$, so the statistical significance of E3's outcome is more complicated to measure.

\section{Further problems with the free-choice paradigm}
In this section, we mention a problem shared by all four of the free-choice experiments considered in the last section; namely, a positive outcome could be blamed on psychological phenomena other than dissonance.  Minimizing the impact of such other phenomena on free-choice experiments has always been of concern to dissonance researchers, going back as far as Brehm's original work~\cite{B}.

For example, suppose that the act of choosing between two objects doesn't change the subject's true ranking, but does force the subject to think more carefully about the true positions of these objects in his/her true ranking.

This behavior can easily be modeled with simulated subjects as in Section 6, using two coin-weightings, $0\leq p\leq P\leq 1$.  The subject's first ranking is subject to the larger volume, $P$, of random noise.  The subject's object choice is subject to the smaller volume, $p$, because s/he is forced to think about it more carefully.  In the second ranking, the positioning of the comparison objects is subject to the smaller volume of noise, because the subject has thought more carefully about these positions.  Since spread depends only on the positioning of the comparison objects, it is equivalent to make the second ranking overall subject to the smaller volume, $p$, of random noise.  Of course the second ranking for a member of the control group of E0 is still subject to the larger volume, $P$.

Experiments E0, E2 and E3 will all report positive outcomes in the presence of this ``think about it more carefully'' phenomenon, even though the subjects never change their true rankings.  To see this without a computer, consider the extreme case where $p=0$ (no random noise) and $P=1$ (the limit case in which rankings are completely random).  In this case, it is straightforward to calculate analytically that the expected spread in E2 and in E3, and the expected ``spread difference'' between the experimental and control groups of E0, are all equal to $16/3$.  However, with milder parameter choices, E2 and E3 are fooled much more dramatically than E0.  For example, with $p=0.5$ and $P=0.9$, the expected spread difference for E0 equals $0.01$, while the expected spread for E3 equals 0.14.  Thus, while all of the experiments are subject to this criticism, when parameters are set to reasonably model human behavior, the flaw appears to be much more problematic in E2 and E3 than in E0.

Memory is another psychological process on which positive results of free-choice experiments might be blamed.  A subject will obviously remember his/her stage-two choice, and might be inclined to construct a stage-three ranking that is consistent with this choice.  Even without assuming any particular coin-flipping model for human behavior, it is obvious that this inclination could lead to positive experimental outcomes in E0, E2 and E3.  For example, compared to the control group in E0, the experimental group will exhibit increased spread caused by this inclination to move the stage-three ranking in the direction of the stage-two choice.  While this might be considered a special case of the choice-induced attitude change that free-choice experiments hope to measure, it is too special of a case because it can be blamed on memory rather than dissonance.

We have not yet discussed E1 because its outcomes will depend on the manner in which true rankings vary from subject to subject, which can't be modeled accurately in computer simulations, since it depends heavily on the particular 15 objects that are used and how real people feel about them.  Nevertheless, it is clear that E1 is subject in principle to the same criticisms as the other experiments.  For example, if each subject's true ranking is modeled as completely random, than E1 is equivalent to E2.

\section{The auxiliary hypotheses of past free-choice experiments}

The probabilistic arguments of Chen and Risen cast some doubt on the conclusions of all past free-choice experiments.  Although E0, E1, E2, and E3 are immune to these criticisms, their results could be blamed on natural psychological processes other than dissonance.  Can any type of free-choice experiment correctly measure choice-induced attitude change caused by dissonance?  One potential glimmer of hope might come from reexamining the auxiliary hypotheses of past free-choice experiments.  Here are summaries of a few of them:
\begin{enumerate}
\item Brehm and Jones performed a free-choice experiment in which subjects rated music albums~\cite{BJ}.  Before choosing between a pair of albums, some of the subjects were told that one of the pair was associated with a record company promotion, so they would receive a free pair of movie tickets if they happened to choose the promoted one (they were not told which one it was).  Other subjects were not told about the promotion until after making their choice.  Among both groups of subjects, some won the free movie tickets and some did not.  The result was that subjects who knew about the promotion before making the choice and who received the movie tickets exhibited less positive spread than subjects in the other three categories.  Presumably, a foreseeable positive benefit to their selection left them with less need to rationalize the choice.

\item Steele, Hopp, and Gonzales used the free-choice paradigm to study the role of affirmation in  dissonance theory~\cite{SHG}.  Some of their subjects valued science highly, while others valued business highly.  In both groups, some of the subjects performed the experiment wearing a lab coat, and some did not.  The result was that subjects who valued science and wore a lab coat exhibited less positive spread than subjects in the other three categories.  One interpretation is to regard dissonance as posing a threat to one's self-system.  The lab coat presumably served as a symbolic affirmation of a core value, thus reducing the need to decrease dissonance through attitude change.

\item Stone used the free-choice paradigm to study the relationship of self-esteem to dissonance~\cite{Stone}.  He found that subjects with high self-esteem showed less positive spread than subjects with low self-esteem, but only when the subjects were primed to think about personal, as opposed to social, standards.

\item  Heine and Lehman found that Japanese men exhibited less positive spread than Canadian men~\cite{HL}.  One interpretation involves the differences between independent and interdependent cultures.  Kitayama, Snibbe, Markus and Suzuki tested this interpretation in a follow-up study in which they asked some of the Japanese men to rank the albums as they thought that \emph{most students} would rank them~\cite{KSMS}.  Subjects given these instructions exhibited larger average spread than subjects who were asked to report personal rankings.  Presumably, dissonance was aroused more by events that were felt to affect interpersonal relationships with others.

\end{enumerate}

Given the known problems with the free-choice paradigm, what should be made of these results?  A significant difference in spread between two groups must mean something.  Since there are now several known possible interpretations of positive spread other than dissonance, one must decide whether any of these alternative interpretations can explain the experimental results.  Chen and Risen noted that positive spread can occur because people maintain true rankings.  Could self-esteem affect the degree to which subjects maintain true rankings?  Positive spread can also occur because of memory or the ``think about it more carefully'' phenomenon.  Does self-esteem correlate with memory?  Could the chance to win a movie ticket make a subject think more carefully about which album has the promotion rather than about his/her true feelings for the albums?  Could a lab coat make a science-valuing subject feel more obliged to construct a second ranking that is consistent with his/her object choice?

Although questions along these lines are important, in some cases the original interpretations of the experimental results feel more plausible than the alternative interpretations.  Also, some of these auxiliary hypotheses were independently tested via the induced-compliance and effort-justification paradigms.  We believe that it is a bit of a stretch to interpret all of the above results under a model in which there is no dissonance.  Thus, taken together, the past free-choice experiments provide at least some evidence of choice-induced attitude change caused by dissonance.

\appendix\section{Expected spread calculation}
In this appendix, we describe the methods by which we derived the expected spread values in Table~\ref{table2}.

In this $12$-object experiment, the subject provides three rankings (the second of which determines his/her stage-two choice).  Thus, there are $(12!)^3\approx 10^{26}$ possible outcomes of this experiment, so the straightforward approach of tabulating all possible outcomes on a computer does not work here.  A statistical approach would be reasonable; for example, one could approximate each expected spread in the table as the average spread of $100,000$ computer simulated runs of the free choice experiment. But this introduces approximation errors.  In the remainder of this section, we instead describe an algorithm that calculates the \emph{exact} value of each expected spread in the table.

The key simplification is to only track the ranking-positions of the two objects between which a stage-two choice is made (say ``toaster'' and ``hairdryer''), treating the other 10 objects as indistinguishable.  So instead of managing all $12!$ permutations of the 12 distinct objects, we will only manage the $12\times 11 = 132$ members of the set, $S$,  of ``simplified rankings'' defined as
$$S=\{(a,b)\mid 1\leq a,b\leq 12, a\neq b\},$$
where $a$ represents the position of the toaster, and $b$ the position of the hairdryer.  For example, the element $(7,3)\in S$ represents the ranking $$(*,*,\text{hairdryer},*,*,*,\text{toaster},*,*,*,*,*).$$  We will consider the members of $S$ to be ordered in some fashion from $1$ to $132$.

We next define a pair of $132$-by-$132$ matrices, called $Q$ and $M$, whose rows and columns will be indexed by the elements of $S$.  For $s_1,s_2\in S$, the entry $Q_{s_1,s_2}$ (in row $s_1$ and column $s_2$ of $Q$) represents the probability that the act of swapping a single randomly chosen adjacent pair of positions will convert $s_1$ into $s_2$.  The entry $M_{s_1,s_2}$ represents the probability that the experiment's random process (repeated swapping randomly chosen adjacent pairs until the coin first lands tails) will convert $s_1$ into $s_2$.  It is straightforward to design an algorithm that explicitly calculates the entries of $Q$.  We can then compute the entries of $M$ in terms of the entries of $Q$ by solving the recursive formula
$$M=(1-p)I + pMQ$$
for $M$, which yields:
$$M=(1-p)(I-pQ)^{-1}.$$
To avoid rounding errors, this matrix calculation could be done over the rational numbers, although the denominators of the entries of $M$ become rather unwieldy.

At first glance, the matrix $M$ might seem useful only for studying a free-choice experiment in which the two objects (``toaster'' and ``hairdryer'') are pre-selected.  However, by exploiting the experiment's symmetry, we will now use $M$ to study the free-choice experiment in which the two object-positions ($i$ and $j$) are pre-selected.

It is convenient to re-name the objects as ``$1$'' through ``$12$'' in the order of the subject's true ranking.  An arbitrary ranking, $\sigma$, can thereby be considered as a bijection of $\{1,...,12\}$, with the subject's true ranking represented by the identity function, $\text{Id}$.  Recall that the subject gives three rankings, $\sigma_1,\sigma_2,\sigma_3$, from which her spread, $\mathfrak{s}=\text{spread}(\sigma_1,\sigma_2,\sigma_3)$ is calculated as:
$$ \mathfrak{s} =
\begin{cases} (j-i)-\left(\sigma_3\circ\sigma_1^{-1}(j)-\sigma_3\circ\sigma_1^{-1}(i)\right) &\mbox{if }
    \sigma_2\circ\sigma_1^{-1}(i)<\sigma_2\circ\sigma_1^{-1}(j),\\
    \left(\sigma_3\circ\sigma_1^{-1}(j)-\sigma_3\circ\sigma_1^{-1}(i)\right) - (j-i) & \mbox{otherwise}, \end{cases}$$

Notice that $\text{spread}(\sigma_1,\sigma_2,\sigma_3)
=\text{spread}(\sigma_1\circ\tau,\sigma_2\circ\tau,\sigma_3\circ\tau)$
for any bijection $\tau$ of $\{1,...,12\}$.  In particular,
$$\mathfrak{s}=\text{spread}(\sigma_1,\sigma_2,\sigma_3)=
\text{spread}(\text{Id},\sigma_2\circ\sigma_1^{-1},\sigma_3\circ\sigma_1^{-1}).$$
Define $s_0=(i,j)\in S$.  In terms of simplified rankings, our list of all possible outcomes of the experiment is:
$$D=\{(s_1,s_2,s_3)\mid s_i\in S\},$$
where $s_1=\sigma_1^{-1}(s_0)$, $s_2=\sigma_2(s_1)$ and $s_3 = \sigma_3(s_1)$, with permutations acting on members of $S$ in the natural way here.
Our strategy is to calculate the spread and the probability of each of the $132^3=2,299,968$ outcomes in $D$, and then compute the expected spread of the experiment as the corresponding weighted sum.

Consider an arbitrary member $\beta=(s_1,s_2,s_3)=((a_1,b_1),(a_2,b_2),(a_3,b_3))\in D$.   The spread of $\beta$ simplifies to:
$$\mathfrak{s}=\begin{cases} (j-i)- (b_3-a_3) &\mbox{if }
    a_2<b_2,\\
    (b_3-a_3)- (j-i) & \mbox{otherwise}. \end{cases}$$
The probability of $\beta$ equals $P_1\cdot P_2\cdot P_3$, where
$$P_1=M_{s_0,s_1},\,\,\, P_2=M_{s_1,s_2},\text{ and }P_3=M_{s_1,s_3}.$$
To justify the first of these three equations, observe that $\sigma_1$ and $\sigma_1^{-1}$ are identically distributed because the probability that any sequence of adjacent swaps will occur is the same as the probability that the sequence will occur in reverse order.

\bibliographystyle{amsplain}
\bibliography{Selinger_Tapp}

\end{document}